\documentclass{llncs}
\let\springervec\vec

\usepackage[utf8]{inputenc}
\usepackage{t1enc,color,array,rotating,graphicx,mathtools,soul}
\usepackage{latexsym,amssymb,amsmath,amsfonts,stmaryrd,algorithm,algpseudocode}
\usepackage{multicol}
\usepackage[english]{babel}
\usepackage[unicode,colorlinks=true]{hyperref}

\numberwithin{equation}{section}
\usepackage{graphicx,nicefrac,wrapfig}
\let\vec\springervec

\makeatletter
\def\@citecolor{blue}%
\def\@urlcolor{blue}%
\def\@linkcolor{blue}%
\def\orcidID#1{\href{http://orcid.org/#1}{\protect\raisebox{-1.25pt}{\protect\includegraphics{
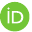%
}}}}
\makeatother

\usepackage[textsize=footnotesize
           ,disable
           ]{todonotes}

\usepackage[capitalize]{cleveref}
\usepackage{multirow,tikz,xspace,comment,enumitem}
\usetikzlibrary{decorations,arrows,arrows.meta,shapes,patterns,automata,calc}

\DeclareFontFamily{U}{mathx}{\hyphenchar\font45}
\DeclareFontShape{U}{mathx}{m}{n}{<-> mathx10}{}
\DeclareSymbolFont{mathx}{U}{mathx}{m}{n}
\DeclareMathAccent{\widebar}{0}{mathx}{"73}

\tikzset{%
  >={Latex[width=2mm,length=2mm]},
          process/.style = {rectangle, draw=black, minimum height=7mm,
                            node distance=2.2cm, text centered,
                            fill=blue!20},
          input/.style={},
          output/.style={text=red},
          highlight/.style={fill=yellow!50},
}

\newcommand{\KK}[2][]{\todo[color=purple!20,#1]{KK: #2}}

\newcommand*\dom{\operatorname{dom}\nolimits}

\newcommand*\ksmt{\texttt{ksmt}\xspace}
\newcommand*\SAT{\textsf{sat}\xspace}
\newcommand*\UNSAT{\textsf{unsat}\xspace}

\newcommand{\NLin}{\mathcal N}
\newcommand{\Lin}{\mathcal L}

\newcommand*\Approx{{\operatorname{approx}}}
\newcommand*\round[1]{\lfloor{#1}\rceil}

\newcommand*\todofb[2][]{\todo[color=cyan!30,tickmarkheight=.2em,#1]{FB: #2}}
\newcommand*\True{\textsf{true}\xspace}
\newcommand*\False{\textsf{false}\xspace}
\newcommand*\nil{\textsf{nil}\xspace}

\newcommand*\interp[1]{\llbracket#1\rrbracket}

\newcommand*\eps{\ensuremath{\epsilon}}
\newcommand*\epsfull[1][\eps]{#1-full\xspace}

\newcommand*\todorev[2][]{\todo[color=black!15]{rev#1: #2}}

\newcommand\Fsat[1][]{\ensuremath{(F^\mathit{sat}_{#1})}}
\newcommand\Funsat{\ensuremath{(F^\mathit{unsat})}}

\newcommand*\reply[3]{{\scriptsize%
    \setlength{\fboxsep}{.25ex}\noindent\colorbox{#2}%
    {\emph{#1:#3}}}}

\newcommand*\fb[1]{\reply{FB}{cyan!30}{#1}}

\newcommand*\kk[1]{\reply{KK}{purple!20}{#1}}
\newcommand*\mk[1]{\reply{MK}{green!20}{#1}}

\newcommand\cade[2]{#2}

\newcommand*\None{\textsf{None}}

\title{
The \ksmt calculus is a $\delta$-complete decision procedure for non-linear constraints\thanks{%
    This research was partially supported by an Intel research grant, the DFG grant WERA MU 1801/5-1 and the RFBR-JSPS 20-51-5000 grant.
}
}
\author{%
    Franz~Brauße\inst{2}\orcidID{0000-0002-2386-7489}
    \and Konstantin~Korovin \inst{2}\orcidID{0000-0002-0740-621X}
    \and Margarita~V.~Korovina\inst{3}\orcidID{0000-0002-2707-0231}
    \and Norbert~Th.~Müller\inst{1}\orcidID{0000-0003-3684-3029}
}
\institute{%
Abteilung Informatikwissenschaften, Universität Trier, Germany \and
The University of Manchester, UK \and
A.P. Ershov Institute of Informatics Systems, Novosibirsk, Russia}

\date{}

\newcommand\irrelcade[1]{}

\begin{document}

\maketitle

\begin{abstract}
    \ksmt is a CDCL-style calculus for solving non-linear constraints over real numbers involving polynomials and transcendental functions. In this paper we investigate properties of the \ksmt calculus and show that it is a $\delta$-complete decision procedure for bounded problems.
    We also propose an extension with local linearisations, 
    which allow for more efficient treatment of non-linear constraints.
\end{abstract}

\section{Introduction}\label{sec:intro}
\strut\KK{Spell check final version!}%
\todo{All references to appendices need to be replaced. a) Upcoming journal version, b) TR on arXiv, c) just drop \fb b \kk b \mk b
Remember to eliminate all words ``appendix''.}%
Solving
non-linear constraints is important in many applications, including verification of
cyber-physical systems
,
software verification
,
proof assistants for mathematics
~\cite{DBLP:books/sp/Platzer18,DBLP:conf/formats/KuratkoR14,https://doi.org/10.1002/rnc.2914,DBLP:conf/fm/BardBD19,DBLP:journals/corr/HalesABDHHKMMNNNOPRSTTTUVZ15,DBLP:conf/fmcad/BrausseKK20}.
Hence there has been a number of approaches for solving non-linear constraints, involving symbolic methods~\cite{DBLP:journals/cca/JovanovicM12,DBLP:conf/cade/MouraP13,DBLP:journals/fmsd/TiwariL16,DBLP:conf/casc/KorovinKS14} as well as numerically inspired ones, in particular for dealing with transcendental functions~\cite{DBLP:conf/cade/GaoAC12,DBLP:conf/cade/TungKO16}, and
combinations of symbolic and numeric methods~\cite{DBLP:conf/frocos/BrausseKKM19,DBLP:journals/tocl/CimattiGIRS18,DBLP:conf/frocos/FontaineOSV17}.

 In \cite{DBLP:conf/frocos/BrausseKKM19} we introduced the \ksmt calculus 
for solving non-linear constraints over a large class of functions including polynomial, exponential and trigonometric functions.
The \ksmt calculus%
\footnote{Implementation is available at \url{http://informatik.uni-trier.de/~brausse/ksmt/}} combines CDCL-style 
reasoning~\cite{DBLP:conf/iccad/SilvaS96,DBLP:conf/vmcai/MouraJ13,DBLP:journals/jar/BonacinaGS20} over reals based on conflict resolution~\cite{CR:KTV09} with incremental linearisations of non-linear 
functions using methods from computable analysis~\cite{DBLP:series/txtcs/Weihrauch00,DBLP:conf/cca/Muller00}. 
Our approach is based on computable analysis and exact real arithmetic which avoids limitations of double precision computations caused by rounding errors and instabilities in numerical methods. In particular, satisfiable and unsatisfiable results returned by \ksmt are exact as required in many applications.
This approach also supports implicit representations of functions as solutions of ODEs and PDEs~\cite{DBLP:books/daglib/0087495}. 

It is 
well known that in the presence of transcendental functions the constraint satisfiability problem is undecidable~\cite{DBLP:journals/jsyml/Richardson68}.
However if we
only require solutions up to some specified precision $\delta$, then the problem can be solved algorithmically on bounded instances and that is the motivation behind $\delta$-completeness, which was introduced in~\cite{DBLP:conf/cade/GaoAC12}.
In essence a $\delta$-complete procedure decides if a formula is
unsatisfiable or a $\delta$ weakening of the formula is satisfiable.

In this paper we investigate theoretical properties of the \ksmt calculus, and its extension $\delta$-\ksmt for the $\delta$-SMT setting.
Our main results are as follows:
\begin{enumerate}
    \item  We introduced a notion of \emph{\epsfull{} linearisations} and prove that all \epsfull{} runs of \ksmt are terminating on bounded instances. 
\item  We extended the $\ksmt$ calculus to the $\delta$-satisfiability setting and proved that $\delta$-\ksmt is a
\emph{$\delta$-complete decision procedure} for bounded instances.
\item We introduced an algorithm for computing \epsfull{} \emph{local linearisations} and integrated it into $\delta$-\ksmt.  Local linearisations can be used to considerably narrow the search space
by taking into account local behaviour of non-linear functions avoiding computationally expensive global analysis.
\end{enumerate}

In \Cref{sec:ksmt-overview}, we give an overview about the \ksmt calculus and
introduce the notion of \epsfull linearisation used throughout the rest of the
paper. We also present a completeness theorem. \Cref{sec:delta-sat} introduces
the notion of $\delta$-completeness and related concepts. In
\Cref{sec:delta-ksmt-cade}
we introduce the $\delta$-\ksmt adaptation, prove it is correct and
$\delta$-complete, and give concrete effective linearisations
based on a uniform modulus of continuity.
Finally in \Cref{sec:eps-from-M-cade}, we introduce local linearisations and show that
termination
is independent of computing uniform moduli of continuity,
before we conclude in \Cref{sec:concl}.

\section{Preliminaries}\label{sec:prelim}
\irrelcade{
\todo[inline]{NM, KK: Reframe notions to intervals which might be easier to digest for target audience. \fb n}
}
The following conventions are used throughout this paper.
By $\Vert\cdot\Vert$ we denote the maximum-norm $\Vert (x_1,x_2,\ldots,x_n)\Vert=\max\{|x_i|:1\leq i\leq n\}$.
When it helps clarity, we write finite and infinite sequences $\vec x=(x_1,\ldots,x_n)$ and $\vec y=(y_i)_i$ in bold typeface.
We are going to use open balls
$B(\vec c,\epsilon)=\{\vec x:\Vert\vec x-\vec c\Vert<\epsilon\}\subseteq\mathbb R^n$
for $\vec c\in\mathbb R^n$ and $\eps>0$ and $\widebar A$ to denote the closure of
the set $A\subseteq\mathbb R^n$ in the
standard topology induced by the norm.
By $\mathbb Q_{>0}$ we denote the set $\{q\in\mathbb Q:q>0\}$.
For sets $X,Y$, a (possibly partial) function from $X$ to $Y$ is written as $X\to Y$.
We use
the notion of compactness: a set $A$ is compact iff every open cover of $A$ has a finite subcover.
\todofb{Note: We use this for $\widebar{D_P}$ in the proof of \Cref{th:int:lin:alt-cade}.}%
In Euclidean spaces this is equivalent to $A$ being bounded and closed~\cite{Willard70}.

\subsection*{Basic notions of Computable Analysis}
\label{sec:ca}

Let us recall the notion of computability of
functions over real numbers used throughout this paper.
A rational number $q$ is an {\it $n$-approximation} of a real number $x$ if $\Vert q-x\Vert\leq2^{-n}$. 
Informally, a function $f$ is {\it computed} by a function-oracle Turing machine
$M_f^?$, where $^?$ is a placeholder for the oracle representing the argument of the function, in the following way.
The real argument $x$ is represented by an oracle function $\varphi:\mathbb N\to\mathbb Q$, for each $n$ returning   
an $n$-approximation $\varphi_n$ of $x$.
For simplicity, we refer to $\varphi$ by the sequence $(\varphi_n)_n$.
When run with argument $p\in\mathbb N$, $M_f^\varphi(p)$ computes a rational
$p$-approximation
of $f(x)$ by querying its oracle $\varphi$ for approximations of $x$.
Let us note that the definition of the oracle machine does not depend on the concrete oracle, i.e., the oracle can be seen as a parameter. In case only the machine without a concrete oracle is of interest, we write $M_f^?$.
We refer to~\cite{DBLP:books/daglib/0067010} for a precise definition of the model of computation by function-oracle Turing machines which is standard in computable analysis.

\begin{definition}[\cite{DBLP:books/daglib/0067010}]\label{def:computable}
Consider $\vec x \in \mathbb R^n$. A \emph{name} for $\vec x$ is a rational sequence
$\vec\varphi=(\vec\varphi_k)_k$
such that $\forall k:\Vert\vec\varphi_k-\vec x\Vert\leq2^{-k}$.
 A function $f:\mathbb R^n\to\mathbb R$ is \emph{computable} iff there is a
    function-oracle Turing machine 
    $M_f^?$
    such that for all $\vec x\in\dom f$ and names  $\vec\varphi$ for $\vec x$,
    $|M_f^{\vec\varphi}(p)-f(\vec x)|\leq2^{-p}$ holds for all $p\in\mathbb N$.
\end{definition}
This definition is closely related to interval arithmetic with unrestricted precision, but enhanced with the guarantee of convergence and it is equivalent to the notion of computability used in~\cite{DBLP:series/txtcs/Weihrauch00}. 
The class of computable functions contains polynomials and transcendental functions like $\sin$, $\cos$, $\exp$, among others.
It is well known \cite{DBLP:books/daglib/0067010,DBLP:series/txtcs/Weihrauch00}
that this class is closed under composition and
that computable functions are continuous.
By continuity, a computable function $f:\mathbb R^n\to\mathbb R$ total on a compact $D\subset\mathbb R^n$
has a computable \emph{uniform modulus of continuity} $\mu_f:\mathbb N\to\mathbb N$ on $D$~\cite[Theorem 6.2.7]{DBLP:series/txtcs/Weihrauch00},
that is,
\begin{align}
    \forall k\in\mathbb N\,\forall\vec y,\vec z\in D:
    \Vert\vec y-\vec z\Vert\leq2^{-\mu(k)}\implies|f(\vec y)-f(\vec z)|\leq2^{-k}
    \text.
    \label{eq:mu}
\end{align}
A uniform modulus of continuity of $f$
expresses how changes in the value of $f$ depend on changes of the arguments
in a uniform way. 

\section{The \ksmt calculus}\label{sec:ksmt-overview}
We first describe the \ksmt calculus
for solving non-linear constraints~\cite{DBLP:conf/frocos/BrausseKKM19} informally, and subsequently recall the main definitions which we use in this paper.
The \ksmt calculus consists of transition rules, which, for any formula in linear separated form, allow deriving lemmas consistent with the formula and, in case of termination,
produce a satisfying assignment for the formula or show that it is unsatisfiable.
A quantifier-free formula is in separated linear form $\mathcal L\cup\mathcal N$ if $\mathcal L$ is a set of clauses over linear constraints and $\mathcal N$ is a set of 
non-linear atomic constraints; this notion is rigorously defined below.

In the \ksmt calculus there are four transition rules applied to its states: Assignment refinement $(A)$, Conflict resolution $(R)$, Backjumping $(B)$ and Linearisation $(L)$.
The final \ksmt states are \SAT and \UNSAT.
A non-final \ksmt state is
a triple $(\alpha,\mathcal L,\mathcal N)$ where $\alpha$ is a (partial) assignment of variables to rationals.
A \ksmt derivation
starts with an initial state where $\alpha$ is
 empty and  tries to extend this assignment to a solution of
 $\mathcal  L \cup \mathcal N$ by
repeatedly
 applying the Assignment refinement rule.
 When such assignment extension is not possible
 we either obtain a linear conflict which is resolved using the  conflict resolution rule,  or a non-linear conflict which is resolved using the linearisation rule.
 
The main idea behind the linearisation rule is to approximate the non-linear constraints around the conflict using linear constraints in such a way that the conflict will be shifted into the linear part where it will be resolved using conflict resolution.
Applying either of these two rules results in a state containing a clause evaluating to \False under the current assignment. They either result in application
of the backjumping rule, which undoes assignments or in termination in case the formula is \UNSAT.
 In this procedure, only the assignment and linear part of the state change and the non-linear part stays fixed. 

\begin{figure}[tp]
\centering\begin{tikzpicture}[node distance=6em,every node/.style={scale=0.74},align=center]
\matrix[row sep=1em,column sep=2em] {
& \node (input) [input] {\clap{separated linear form}};
\\
\\
& \node (lin-check) [draw,thick,star,star points=9,star point ratio=0.87,inner sep=.5ex,fill=yellow!40] {lin. \\ check};
\\ \node (choice) [draw,thick,diamond,inner sep=.5ex,fill=green!30] {choice};
&
& \node (A) [process,thick] {A};
&
& \node (B) [process,very thick,draw=red] {B};
\\
&
& \node (ex-z) [draw,red,text=black,very thick,star,star points=9,star point ratio=0.87,inner sep=.5ex,fill=yellow!40] {$\exists z$};
& \node (R) [process,thick] {R}; \\
& \node (nlin-check) [draw,star,thick,star points=9,star point ratio=0.87,inner sep=.5ex,fill=yellow!40] {nlin. \\ check};
&
& \node (L) [process,thick] {L};
\\
};
\draw[->] (input) -- node [right] {$\alpha=\nil$} (lin-check);
\draw[->] (lin-check) edge [in=90,out=180] node [above left] {p.lin.cons.} (choice);
\draw[->] (lin-check) edge [in=90,out=0] node [above right] {p.lin.incons.} (B);
\draw[->] (choice) edge [out=-45,in=180] (ex-z);
\draw[->] (choice) edge [out=270,in=180] (nlin-check);
\draw[->] (A) edge [bend right=10] (choice.east);
\draw[->] (B) edge [bend right=20] (choice);
\draw[->] (ex-z) -- node [left,midway,yshift=-1ex] {$\exists q$} (A);
\draw[->] (ex-z) -- node [above,near start] {$\neg\exists q$} (R);
\draw[->] (R) edge [bend right=20] (B);
\draw[->] (nlin-check) edge [bend right=15] node [right,near start,yshift=-.5ex] {p.nlin.cons.} (ex-z);
\draw[->] (nlin-check) edge [bend right=15] node [below] {p.nlin.incons.} (L);
\draw[->] (L) edge [bend right] (B);
\end{tikzpicture}
\caption{Core of \ksmt calculus. Derivations terminate in 
red nodes.}
\label{fig:white-box}
\end{figure}
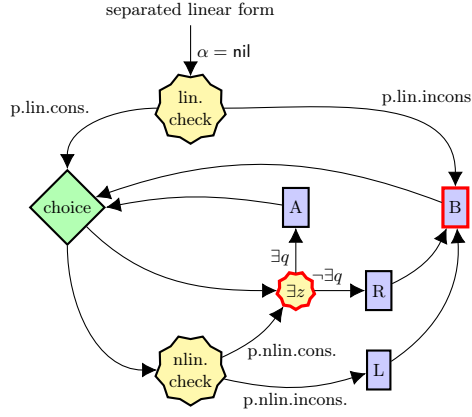

\paragraph{Notations.}
Let
$\mathcal F_{\mathrm{lin}}$ consist of rational constants, addition and multiplication
by rational constants; $\mathcal F_{\mathrm{nl}}$ denotes an arbitrary collection of non-linear
computable
functions including transcendental functions and polynomials over the reals.
We consider the structure
$(\mathbb R,\langle\mathcal F_{\mathrm{lin}}\cup\mathcal F_{\mathrm{nl}},\mathcal P\rangle)$
where $\mathcal P=\{{<},{\leq},{>},{\geq},{=},{\neq}\}$ and a set of variables $V=\{x_1,x_2,\ldots,x_n,\ldots\}$.
We will use, possibly with indices, $x$ to denote variables and
$q,c,e$ for rational constants.
Define terms, predicates and formulas over $V$ in the standard way.
An \emph{atomic linear constraint} is a formula of the form: 
$q+c_1x_1+\ldots+c_nx_n \diamond 0$ where $q,c_1,\ldots,c_n\in\mathbb Q$ and $\diamond\in\mathcal P$. Negations of atomic formulas can be eliminated by rewriting the predicate symbol $\diamond$ in the standard way, hence we assume that all literals are positive.
A \emph{linear constraint} is a disjunction of atomic linear constraints, also called \emph{(linear) clause}.
An \emph{atomic non-linear} constraint is a formula of the form 
$f(\vec x)\diamond 0$, where $\diamond\in\mathcal P$
and $f$ is a composition of computable non-linear functions from $\mathcal F_{\mathrm{nl}}$ over variables $\vec x$.
Throughout this paper
for every computable real function $f$
we use $M_f^?$ to denote a function-oracle Turing machine computing $f$.
We assume quantifier-free formulas in \emph{separated linear
form}~\cite[Definition~1]{DBLP:conf/frocos/BrausseKKM19},
that is, $\mathcal L\cup\mathcal N$ where $\mathcal L$ is a set of linear constraints and $\mathcal N$ is a set of non-linear atomic constraints. 
Arbitrary quantifier-free formulas can be transformed equi-satisfiably into separated linear form in polynomial time~\cite[Lemma~1]{DBLP:conf/frocos/BrausseKKM19}.
Since in separated linear form all non-linear constraints are atomic we will call them just \emph{non-linear constraints}.

Let $\alpha:V\to\mathbb Q$ be a partial variable assignment.
The interpretation $\interp{\vec x}^\alpha$ of a vector of variables $\vec x$ under $\alpha$ is defined in a standard way as component-wise application of $\alpha$.
Define the notation $\interp{t}^\alpha$ as
evaluation of term $t$ under assignment $\alpha$, that can be partial,
in which case
$\interp{t}^\alpha$ is treated symbolically.
We extend $\interp\cdot^\alpha$ to predicates, clauses and CNF in the
usual way and $\True,\False$ denote the constants of the Boolean domain.
The evaluation $\interp{t\diamond 0}^\alpha$ for a predicate $\diamond$ and a term $t$ results in $\True$ or $\False$ only if all variables in $t$ are assigned by $\alpha$.

In order to formally restate the calculus, the notions of linear resolvent and linearisation are essential.
A resolvent $R_{\alpha,\mathcal L,z}$ on a variable $z$ is a set of linear constraints 
that do not contain $z$, are implied by the formula $\mathcal L$ and which evaluate to $\False$ under the current partial assignment $\alpha$; for more
details see~\cite{CR:KTV09,DBLP:conf/frocos/BrausseKKM19}.

\begin{definition}\label{def:linearisation}
Let $P$ be a non-linear constraint and let $\alpha$ be an 
assignment with $\interp P^\alpha=\False$.
A \emph{linearisation of $P$ at $\alpha$} is a linear clause $C$
with the properties:
\begin{enumerate}
\item\label{def:linearisation:it1}
    $\forall\beta:\interp P^\beta=\True\implies\interp C^\beta=\True$,
    and
\item\label{def:linearisation:it2}
    $\interp C^\alpha=\False$.
\end{enumerate}
\end{definition}
Wlog.\ we can assume that the variables of $C$ are a subset of the variables of $P$.
Let us note that any linear clause $C$ represents the complement of a rational polytope $R$ and we will use both
interchangeably. Thus for a rational polytope $R$,  $\vec{x}\not \in R$ also stands for a linear clause. 
In particular, any linearisation excludes a rational polytope containing the conflicting assignment from the search space. 

\paragraph{Transition rules.}
For a formula $\mathcal L_0\cup\mathcal N$ in separated linear form,
the initial \ksmt state is $(\nil,\mathcal L_0,\mathcal N)$. The calculus consists of the following transition rules from a state $S=(\alpha,\mathcal L,\mathcal N)$ to $S'$:
\begin{description}
\item[$(A)$]\label{rule:A}
    \emph{Assignment.}
    $S'=(\alpha::z\mapsto q,\mathcal L,\mathcal N)$ iff $\interp{\mathcal L}^\alpha\neq\False$ and there is a variable $z$ unassigned in $\alpha$ and $q\in\mathbb Q$ with $\interp{\mathcal L}^{\alpha::z\mapsto q}\neq\False$.
\item[$(R)$]\label{rule:R'}
    \emph{Resolution.}
    $S'=(\alpha,\mathcal L\cup R_{\alpha,\mathcal L,z},\mathcal N)$ iff
    $\interp{\mathcal L}^\alpha\neq\False$ and there is a variable $z$ unassigned in $\alpha$ with $\forall q\in\mathbb Q:\interp{\mathcal L}^{\alpha::z\mapsto q}=\False$
    and $R_{\alpha,\mathcal L,z}$ is a resolvent.
\item[$(B)$]\label{rule:B}
    \emph{Backjump.}
    $S'=(\gamma,\mathcal L,\mathcal N)$ iff $\interp{\mathcal L}^\alpha=\False$ and there is a maximal prefix $\gamma$ of $\alpha$ such that $\interp{\mathcal L}^\gamma\neq\False$.
\item[$(L)$]\label{rule:L}
    \emph{Linearisation.}
    $S'=(\alpha,\mathcal L\cup \{L_{\alpha, P}\},\mathcal N)$ iff
    $\interp{\mathcal L}^\alpha\neq\False$,
    there is $P$ in $\mathcal N$ with $\interp P^\alpha=\False$ and
    there is a linearisation
    $L_{\alpha,P}$ of $P$ at $\alpha$.
\item[\Fsat]
    \emph{Final \SAT.}
    \todo{prefer $\SAT_\alpha$ \kk y \mk y}%
    $S'=\SAT$ if all variables are assigned in $\alpha$,
    $\interp{\mathcal L}^\alpha=\True$ and
    none of the rules $(A),(R),(B),(L)$ is applicable.
\item[\Funsat]
    \emph{Final \UNSAT.}
    $S'=\UNSAT$ if $\interp{\mathcal L}^\nil=\False$. In other words a trivial contradiction, e.g., $0>1$ is in $\mathcal L$.
\end{description}

A path (or a run) is a derivation in a \ksmt.
A procedure is an effective (possibly non-deterministic) way to construct a path.

\paragraph{Termination.}\label{par:termination}
If no transition rule is applicable, the derivation terminates.
For clarity, we added the explicit rules $\Fsat$ and $\Funsat$ which lead to the final states. 
This calculus is sound \cite[Lemma~2]{DBLP:conf/frocos/BrausseKKM19}: if the final transition is $\Fsat$, then $\alpha$ is a solution to the original formula, or $\Funsat$, then a trivial contradiction $0>1$ was derived and  the original formula is unsatisfiable. The calculus also makes progress by reducing the search space~\cite[Lemma~3]{DBLP:conf/frocos/BrausseKKM19}.

\begin{figure}
\centering
\small
\begin{minipage}{.47\linewidth}
\[ \mathcal C=
   \underbrace{\textcolor{red}{(y\leq 1/x)}}_P{}
   \begin{aligned}[t]
	&{}\land{}\textcolor{green!40!black}{(x/4+1\leq y)} \\
	&{}\land\textcolor{green!40!black}{(y\leq 4\cdot(x-1))} \\
	&{{}\land\big((x\leq\tfrac{12}{19})\lor(y\leq\tfrac{19}{12})\big)} \\
	&{{}\land\textcolor{blue}{\big((x\leq\tfrac{220}{223})\lor(y\leq\tfrac{223}{220})\big)}} \\
	&{{}\land\textcolor{purple}{(\tfrac43\leq x)\land(x\leq\tfrac{220}{223})}} \\
	&{{}\land\textcolor{cyan!50!black}{(\tfrac43\leq\tfrac{220}{223})}}
   \end{aligned} \]

\centering
\includegraphics[width=.8\linewidth]{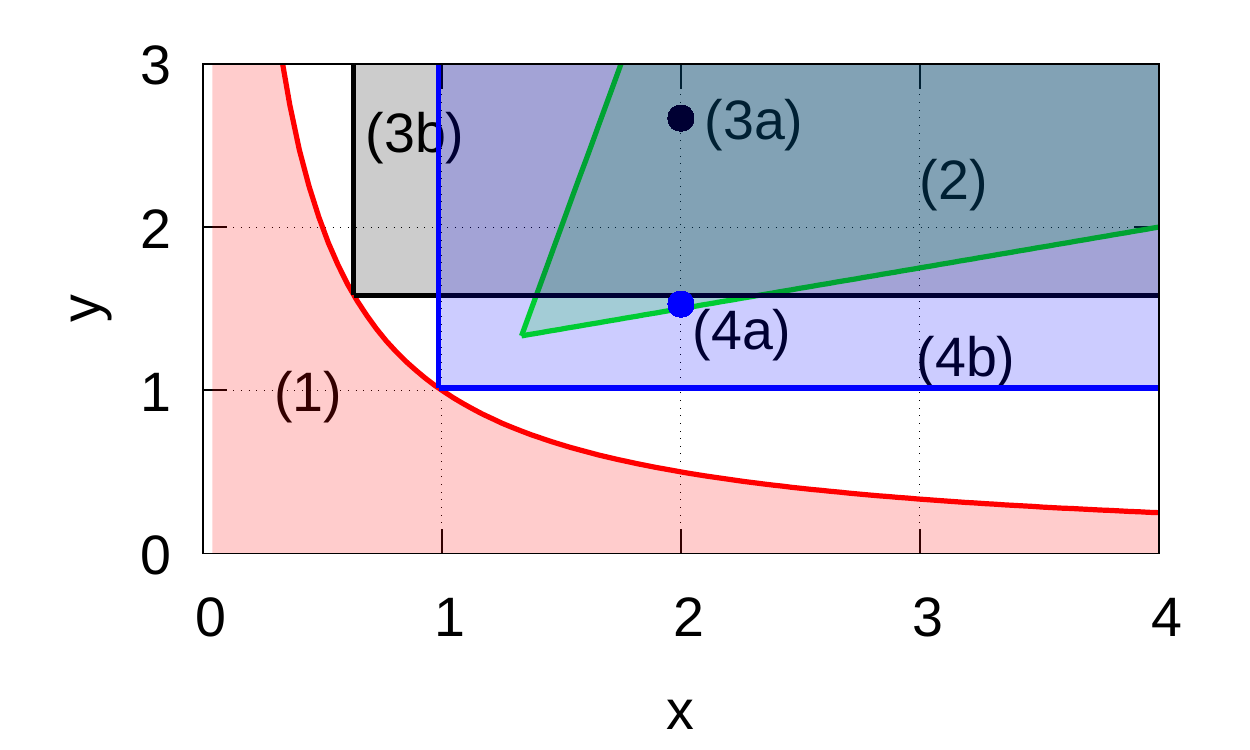}
\end{minipage}%
\begin{minipage}{.53\linewidth}\small
Linearisation of $P$ on conflicts $(x,y)$ at $\alpha$ here:
\begin{itemize}
\vspace{-.5\baselineskip}
\item choose $d\coloneqq(1/\interp x^\alpha+\interp y^\alpha)/2$,
\item 
	$C=\big(x\leq 1/d\;\lor\; y\leq d\big)$
\end{itemize}
\centering
\begin{tabular}{c|l|c}\removelastskip
	rule & $\alpha$ & note \\ \hline
	$(A)$ & $x\mapsto 2$ & \\
	$(A)$ & $x\mapsto 2$, $y\mapsto \tfrac83$ & (3a) \\
	$(L)$ & $x\mapsto 2$, $y\mapsto \tfrac83$ & (3b) \\
	$(B)$ & $x\mapsto 2$ & \\
	$(A)$ & $x\mapsto 2$, $y\mapsto\tfrac{84}{55}$ & (4a) \\
	$(L)$ & $x\mapsto 2$, $y\mapsto\tfrac{84}{55}$ & (4b) \\
	$(B)$ & $x\mapsto 2$ & \\
	$(R)$ & $x\mapsto 2$ & \textcolor{green!40!black}{on} \textcolor{blue}{$y$} \\
	$(B)$ & & \\
	$(R)$ & & \textcolor{purple}{on $x$} \\
	\Funsat && \texttt{unsat}
\end{tabular}
\end{minipage}
\caption{\texttt{unsat} example run of \ksmt using interval linearisation~\cite{DBLP:conf/frocos/BrausseKKM19}.}
\label{fig:lin-exampe-unsat-run}
\end{figure}

An example run of the \ksmt calculus is presented in \Cref{fig:lin-exampe-unsat-run}. We start in a state with a non-linear part $\NLin=\{y\leq 1/x\}$, which defines the pink area and the linear part $\Lin=\{ (x/4+1\leq y), (y\leq4\cdot(x-1))\}$,  shaded in green. Then we successively apply \ksmt rules excluding regions around candidate solutions by
linearisations,
until we derive linearisations 
which separates the pink area from the green area thus deriving a contradiction.

\begin{remark}\label{rem:lin-infty-often}
In general a derivation may not terminate.
The only cause of non-termination is 
the linearisation rule which adds new linear constraints and can
be applied infinitely many times.
To see this, observe that \ksmt with only
the rules $(A),(R),(B)$ corresponds to the conflict resolution
calculus which is known to be terminating~\cite{CR:KTV09,DBLP:conf/cade/KorovinV11}. Thus, in infinite \ksmt runs the linearisation rule $(L)$ is applied infinitely often. This argument is used in the proof of \Cref{th:bounded} below.
Let us note that during a run the \ksmt calculus neither conflicts nor lemmas can be generated more than once. In fact, any generated linearisation is not implied by the linear part, 
prior to adding this linearisation.
\end{remark}

\subsection{Sufficient termination conditions}
In this section we will assume that $(\alpha,\mathcal L,\mathcal N)$ is a \ksmt state obtained by applying \ksmt inference rules to an initial state.  
As in \cite{DBLP:conf/cade/GaoAC12} we only consider bounded instances. In many applications this is a natural assumption as variables usually range within some (possibly large) bounds.  
\irrelcade{\todofb{remove sentence, implied by `bounded set'}}%
We can assume that these bounds are  made explicit as linear constraints in the system.

\begin{definition}
Let $F$ be the formula $\mathcal L_0\land\mathcal N$ in separated linear form over variables $x_1,\ldots,x_n$ and
let $B_i$ be the set defined by the conjunction of all clauses in $\mathcal L_0$ univariate in $x_i$, 
for $i=1,\ldots,n$; in particular, if there are no univariate linear constraints over $x_i$ then $B_i=\mathbb R$.
We call $F$ a \emph{bounded instance} if:
\begin{itemize}
\item
$D_F\coloneqq\bigtimes_{i=1}^n B_i$ is bounded,
and
\item for each non-linear constraint $P:f(x_{i_1},\ldots,x_{i_k})\diamond 0$ in $\mathcal N$
with $i_j\in\{1,\ldots,n\}$ for $j\in\{1,\ldots,k\}$
it holds that
$\widebar{D_P}\subseteq\dom f$ where
$D_P\coloneqq\bigtimes_{j=1}^k B_{i_j}$.
\end{itemize}
\end{definition}
By this definition, already the linear part of bounded instances
explicitly defines a bounded set by univariate constraints.
Consequently, the set of solutions of $F$ is bounded as well.

In \Cref{th:bounded} we show that when we consider bounded instances and restrict linearisations to so-called $\epsilon$-full linearisations, then the procedure terminates. We use this to show that the \ksmt-based decision procedure we introduce in \Cref{sec:delta-ksmt-cade} is $\delta$-complete.
\begin{definition}\label{def:eps-full}
Let $\eps>0$, $P$ be a non-linear constraint over variables $\vec x$ and let
$\alpha$ be an assignment of $\vec x$.
A linearisation $C$ of $P$ at $\alpha$ is called \emph{\epsfull} iff
for all assignments $\beta$ of $\vec x$ with
$\interp{\vec x}^\beta\in B(\interp{\vec x}^\alpha,\eps)$, $\interp{C}^\beta=\False$.

A \ksmt 
\todofb{`almost all' weakens \Cref{th:int:lin-epsfull-cade,th:int:lin:alt-cade}.}%
run is called \epsfull for some $\eps>0$, if
all but finitely many linearisations in this run are $\epsilon$-full.
\end{definition}

The next theorem provides a basis for termination of \ksmt-based decision procedures for satisfiability.
\begin{theorem}\label{th:bounded}
Let $\eps>0$.
On bounded instances,
\epsfull \ksmt runs
are terminating.
\end{theorem}
\begin{proof}
Let $F:\mathcal L_0 \wedge \mathcal N$ be a bounded instance and $\eps>0$.
Towards a contradiction assume
there is an infinite
\epsfull
derivation $(\alpha_0,\mathcal L_0,\mathcal N),\dots, (\alpha_n,\mathcal L_n,\mathcal N), \dots $ in the \ksmt calculus.
Then, by definition of the transition rules, $\mathcal L_k\subseteq\mathcal L_l$ for all $k,l$ with $0\leq k\leq l$.
According to \Cref{rem:lin-infty-often} in any infinite derivation the linearisation rule must be applied infinitely many times.
During any run of \ksmt the set of non-linear constraints $\mathcal N$ is fixed and therefore
there is a non-linear
constraint $P$ in $\mathcal N$ over variables $\vec x$ to which linearisation is applied infinitely often. Let 
$(\alpha_{i_1},\mathcal L_{i_1},\mathcal N),\dots, (\alpha_{i_n},\mathcal L_{i_n},\mathcal N), \dots$ be a corresponding subsequence in the derivation such that
$C_{i_1}\in \mathcal L_{i_1+1},\ldots,C_{i_n}\in \mathcal L_{i_n+1},\ldots$ are $\epsilon$-full linearisations of $P$.
Consider
two different linearisation steps $k,\ell\in\{i_j:j\in\mathbb N\}$
in the derivation
where $k < \ell$.
By the precondition $\interp{\mathcal L_\ell}^{\alpha_\ell}\neq\False$ of rule $(L)$ applied in step $\ell$, in particular the linearisation $C_k\in\mathcal L_{k+1}\subseteq\mathcal L_\ell$ of $P$ constructed in step $k$ does not evaluate to \False under $\alpha_\ell$.
Since the set of variables in $C_k$ is a subset of those in $P$, $\interp{C_k}^{\alpha_\ell}\neq\False$ implies $\interp{C_k}^{\alpha_\ell}=\True$.
By assumption, the linearisation $C_k$ is \epsfull, thus
from Definition~\ref{def:eps-full} it follows that
$\interp{\vec x}^{\alpha_\ell}\notin B(\interp{\vec x}^{\alpha_k},\epsilon)$.
Therefore the distance between $\interp{\vec x}^{\alpha_k}$ and $\interp{\vec x}^{\alpha_\ell}$ is at least $\epsilon$.  
However, every conflict satisfies the variable bounds defining $D_F$, so there could be only finitely many conflicts with pairwise distance at least $\epsilon$. 
This contradicts the above. 
\end{proof}

Concrete algorithms to compute \epsfull linearisations are presented in \Cref{sec:delta-ksmt-cade,sec:eps-from-M-cade}.

\section{$\delta$-decidability}\label{sec:delta-sat}
In the last section, we proved termination of the \ksmt calculus on bounded instances
when linearisations are \epsfull.
Let us now investigate how \epsfull linearisations of
constraints involving non-linear computable functions can be constructed.
To that end, we assume that all non-linear functions are defined on the closure of the bounded space $D_F$ defined by the bounded instance $F$.

So far we described an
approach which gives exact results but at the same time is necessarily incomplete due to undecidability of non-linear constraints in general. 
On the other hand, non-linear constraints usually can be approximated using numerical methods allowing to obtain approximate solutions to the problem.
This gives rise to the 
bounded $\delta$-SMT problem~\cite{DBLP:conf/cade/GaoAC12} which allows an overlap between the properties $\delta$-\SAT and \UNSAT of formulas
as illustrated by \Cref{fig:lin-mv-cade}.
It is precisely this overlap that enables $\delta$-decidability of bounded instances.

\irrelcade{
The notion of computability for functions over the reals is \todofb{cite}%
closely related to effective continuity.
That way, the problem of deciding satisfiability of a formula can be lifted to a continuous domain where functions are no longer handled on a purely algebraic basis only as symbols, but also gain the quality of approximability.
In particular, computation of these functions is performed on continuous spaces of sequences of approximations of real numbers.
}

Let us recall the notion of $\delta$-decidability, adapted from~\cite{DBLP:conf/cade/GaoAC12}.
\begin{definition}\label{def:pred-approx}%
Let $F$ be a formula in separated linear form and
let $\delta\in\mathbb Q_{>0}$.
We inductively define the $\delta$-weakening $F_\delta$ of $F$.
\begin{itemize}
\item
    If $F$ is linear, let $F_\delta\coloneqq F$.
\item
    If $F$ is a non-linear constraint $f(\vec x)\diamond 0$, let
\[
    F_\delta\coloneqq\begin{cases}
        f(\vec x)-\delta\diamond 0,&\text{if}~\diamond\in\{<,\leq\} \\
        f(\vec x)+\delta\diamond 0,&\text{if}~\diamond\in\{>,\geq\} \\
        |f(\vec x)|-\delta\leq 0,&\text{if}~\diamond\in\{=\} \\
        (f(\vec x)<0\lor f(\vec x)>0)_\delta,&\text{if}~\diamond\in\{\neq\}\text.
    \end{cases}
\]
\item
    Otherwise, $F$ is $A\circ B$ with $\circ\in\{\land,\lor\}$. Let $F_\delta\coloneqq(A_\delta\circ B_\delta)$.
\end{itemize}

\noindent\emph{$\delta$-deciding} $F$ designates
computing
\[ 
    \begin{cases}
    \text{\UNSAT},&\text{if}~\interp F^\alpha=\False~\text{for all}~\alpha \\
    \text{$\delta$-\SAT},&\text{if}~\interp{F_\delta}^\alpha=\True~\text{for some}~\alpha\text.
\end{cases} \]
In case both answers are valid, the algorithm may output any.

An assignment $\alpha$ with $\interp{F_\delta}^\alpha=\True$ we call a \emph{$\delta$-satisfying assignment} for $F$.
\end{definition}
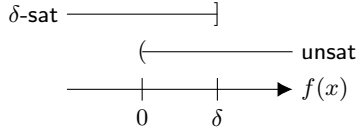
\begin{figure}[tp]
\centering
\begin{tikzpicture}
\draw[->] (-1,0) -- (2,0) node [right] {$f(x)$};
\draw (0,-.15) node [below] {$0$} -- (0,.15);
\draw (1,-.15) node [below] {$\delta$} -- (1,.15);
\draw (-1,1) node [left] {$\delta$-\SAT} -- (1,1) node {$]$};
\draw (0,.5) node {$($} -- (2,.5) node [right] {\UNSAT};
\end{tikzpicture}
\caption{%
    The overlapping cases in the $\delta$-SMT problem $f(x)\leq 0$.
}
\label{fig:lin-mv-cade}
\end{figure}
For non-linear constraints $P$ this definition of the $\delta$-weakening $P_\delta$ corresponds exactly to the notion of $\delta$-weakening $P^{-\delta}$ used in the introduction of $\delta$-decidability \cite[Definition 4.1]{DBLP:conf/lics/GaoAC12}.

\begin{remark}\label{rem:delta:eq-le-ge}
The $\delta$-weakening of a non-linear constraint $f(\vec x)\neq 0$ is a tautology.
\end{remark}

We now consider the problem of $\delta$-deciding quantifier-free formulas
in separated linear form.
The notion of $\delta$-decidability is slightly stronger than in \cite{DBLP:conf/cade/GaoAC12}
in the sense that we do not weaken linear constraints.
Consider a formula $F$ in separated linear form.
As before, we assume variables $\vec x$ to be bounded by linear constraints $\vec x\in D_F$.
We additionally assume that for all non-linear constraints
$P:f(\vec x)\diamond0$ in $\mathcal N$, $f$ is defined on
$\widebar{D_P}$ and, 
in order to simplify the presentation, throughout the rest of paper we will assume only the predicates $\diamond\in\{{>},{\geq}\}$ are part of formulas, since the remaining ones ${<},{\leq},{=}$
can easily be expressed by the former using
simple arithmetic transformations,
and by \Cref{rem:delta:eq-le-ge} predicates $\neq$ are irrelevant for $\delta$-deciding formulas.

An algorithm is \emph{$\delta$-complete}, if it $\delta$-decides
bounded instances%
~\cite{DBLP:conf/cade/GaoAC12}.

\section{$\delta$-\ksmt}\label{sec:delta-ksmt-cade}

Since $\delta$-decidability as introduced above
adapts the condition when a formula is considered to be satisfied to $\delta$-\SAT, this condition has to be reflected in the calculus, which we show solves the bounded $\delta$-SMT problem in this section.
Adding the following
rule $\Fsat[\delta]$ together with the new final state $\delta$-\SAT to \ksmt relaxes the
termination conditions
and turns it into the extended calculus we call $\delta$-\ksmt.
\begin{description}
\item[\ensuremath{\Fsat[\delta]}]
    \emph{Final $\delta$-\SAT.}
    If $(\alpha,\mathcal L,\mathcal N)$ is a $\delta$-\ksmt state
    where $\alpha$ is a total assignment
    and $\interp{\mathcal L\land\mathcal N_\delta}^\alpha=\True$, transition to the $\delta$-\SAT state.
\end{description}
The applicability conditions on the rules $(L)$ and $\Fsat[\delta]$ individually are not decidable~\cite{DBLP:journals/jsyml/Richardson68,Brattka2008},
however, when we compute them simultaneously, we can effectively apply one of these rules,
as we will show in \Cref{th:int:lin-correct-cade}.
In combination with \epsfull{}ness of the computed linearisations (\Cref{th:int:lin-epsfull-cade}),
this leads to \Cref{thm:delta-ksmt-complete-cade},
showing that $\delta$-\ksmt is a $\delta$-complete decision procedure.

Let us note that if we assume $\delta=0$ then $\delta$-\ksmt would just reduce to \ksmt
as $\Fsat$ and $\Fsat[\delta]$ become indistinguishable, but in the following we always assume $\delta>0$.

In the following sub-section, we prove that terminating derivations of the $\delta$-\ksmt calculus lead to correct results.
Then, in \Cref{sec:pf-snd-thm-cade},
we present a concrete algorithm
for applying rules $(L)$ and $\Fsat[\delta]$ and show its linearisations to be \epsfull,
which is sufficient to ensure termination, as shown in \Cref{th:bounded}.
These properties lead to a $\delta$-complete decision procedure.
In \Cref{sec:eps-from-M-cade} we develop a more  practical algorithm for $\epsilon$-full linearisations that does not require computing a uniform modulus of continuity.

\subsection{Soundness}
In this section we show soundness of the $\delta$-\ksmt calculus, that is,
validity of its derivations.
In particular, this implies that
derivability of the final states $\UNSAT$, $\delta$-\SAT and \SAT directly corresponds to unsatisfiability, $\delta$-satisfiability and satisfiability of the original formula, respectively.

\begin{lemma}\label{lem:preserve_assign-cade}
For all $\delta$-\ksmt derivations of $S'=(\alpha',\mathcal L',\mathcal N)$ from a state $S=(\alpha,\mathcal L,\mathcal N)$
and for all total assignments $\beta$,
$\interp{\mathcal L\land\mathcal N}^\beta=
\interp{\mathcal L'\land\mathcal N}^\beta$.
\end{lemma}
\begin{proof}
    Let $\beta$ be a total assignment of the variables in $\mathcal L\land\mathcal N$.
    Since the set of variables remains unchanged by $\delta$-\ksmt derivations,
    $\beta$ is a total assignment for $\mathcal L'\land\mathcal N$ as well.
    Let $S'=(\alpha',\mathcal L',\mathcal N)$ be derived
    from $S=(\alpha,\mathcal L,\mathcal N)$
    by a single application of one of $\delta$-\ksmt rules.
    By the structure of $S'$, its derivation was not caused by
    neither $\Funsat,\Fsat$ or $\Fsat[\delta]$.
    For rules $(A)$ and $(B)$ there is nothing to show since
    $\Lin=\Lin'$.
    If $(R)$ caused $S\mapsto S'$, the claim holds by
    soundness of arithmetical
    resolution. Otherwise $(L)$ caused
    $S\mapsto S'$ in which
    case
    the direction $\Rightarrow$ follows from the definition of a linearisation (condition~\ref{def:linearisation:it1} in \Cref{def:linearisation})
    while the other direction trivially holds since $\Lin\subseteq\Lin'$.

    The condition on derivations of arbitrary lengths then follows by induction.
\end{proof}

\begin{lemma}\label{thm:delta_sound-cade''}
Let $\delta\in\mathbb Q_{>0}$.
Consider a formula $G=\mathcal L_0\land\mathcal N$ in separated linear form and
let $S=(\alpha,\mathcal L,\mathcal N)$ be a $\delta$-\ksmt state derivable from the initial state $S_0=(\nil,\mathcal L_0,\mathcal N)$.
The following hold.
\begin{itemize}
\item
    If rule $\Funsat$ is applicable to $S$
    then $G$ is unsatisfiable.
\item
    If rule $\Fsat[\delta]$ is applicable to $S$
    then $\alpha$ is a $\delta$-satisfying assignment for $G$, hence $G$ is $\delta$-satisfiable.
\item
    If rule $\Fsat$ is applicable to $S$
    then $\alpha$ is a satisfying assignment for $G$,
    hence $G$ is satisfiable.
\end{itemize}
\end{lemma}

\begin{proof}
Let formula $G$ and states $S_0,S$ be as in the premise.
As $S$ is not final in $\delta$-\ksmt, only \ksmt rules have been applied in deriving it. The statements for rules $\Funsat$ and $\Fsat$ thus hold by soundness of \ksmt~\cite[Lemma~2]{DBLP:conf/frocos/BrausseKKM19}.

Assume $\Fsat[\delta]$ is applicable to $S$,
that is, $\interp{\mathcal L\land\mathcal N_\delta}^\alpha$ is \True.
Then, since $\mathcal L_0\subseteq\mathcal L$, we conclude that $\alpha$ satisfies $\mathcal L_0\land\mathcal N_\delta$ which, according to \Cref{def:pred-approx}, equals $G_\delta$. Therefore $\alpha$ is a $\delta$-satisfying assignment for $G$.
\end{proof}
Since the only way to derive one of the final states \UNSAT, $\delta$-\SAT and \SAT from the initial state in $\delta$-\ksmt is by application of the rule $\Funsat,\Fsat[\delta]$ and $\Fsat$,
respectively,
as corollary of \Cref{lem:preserve_assign-cade,thm:delta_sound-cade''} we obtain soundness.
\begin{theorem}[Soundness]\label{thm:delta_sound-cade}
Let $\delta\in\mathbb Q_{>0}$. The $\delta$-\ksmt calculus is sound.
\end{theorem}

\subsection{$\delta$-completeness}\label{sec:pf-snd-thm-cade}

We proceed by
introducing \Cref{alg:box-linearisation'-cade}
computing linearisations and
deciding which of the rules $\Fsat[\delta]$ and $(L)$ to apply.
These linearisations are then shown to be \epsfull for some $\eps>0$ depending on the bounded instance.
By \Cref{th:bounded}, this property implies termination, showing that $\delta$-\ksmt is a $\delta$-complete decision procedure.

Given a non-final $\delta$-\ksmt state,
the function \Call{nlinStep$_\delta$}{} in \Cref{alg:box-linearisation'-cade}
computes a $\delta$-\ksmt state derivable from it by application of $\Fsat[\delta]$ or $(L)$. This is done by evaluating the non-linear functions and adding a linearisation $\ell$ based on their uniform moduli of continuity as needed.
To simplify the algorithm, it assumes 
total assignments as input.
It is possible to relax this requirement, e.g., by invoking rules $(A)$ or $(R)$ instead of returning $\delta$-\SAT for partial assignments.

\begin{algorithm}[tp]
\setlength{\columnsep}{0cm}
\begin{multicols}{2}
\begin{algorithmic}
    \Function{linearise$_\delta$}{$f,\vec x,\diamond,\alpha$}
        \State compute 
            $p\geq-\lfloor\log_2(\min\{1,\delta/4\})\rfloor$
        \State $\varphi\gets(n\mapsto\interp{\vec x}^\alpha)$
        \State $\eps\gets2^{-\mu_f(p)}$
        \State $\tilde y\gets M_f^\varphi(p)$
        \If{$\tilde y\diamond-\delta/2$}
            \State\Return $\None$
        \EndIf
        \State\Return $(\vec x\notin B(\interp{\vec x}^\alpha, \eps))$
    \EndFunction
\columnbreak
\Function{nlinStep$_\delta$}{$\alpha,\mathcal L,\mathcal N$}
    \For{$P:(f(\vec x)\diamond 0)$ \textbf{in} $\mathcal N$}
        \State $\ell\gets{}$\Call{linearise$_\delta$}{$f,\vec x,\diamond,\alpha$}
        \If{$\ell\neq\None
        $}
            \State\Return $(\alpha,\mathcal L\cup\{\ell\},\mathcal N)$
            \Comment 
                $(L)$
        \EndIf
    \EndFor
    \State \Return $\delta$-\SAT
    \Comment 
        $\Fsat[\delta]$
\EndFunction
\end{algorithmic}
\end{multicols}
\caption{%
    (\textsc{nlinStep$_\delta$})
    Algorithm computing a $\delta$-\ksmt derivation
    according to either rule $(L)$ or $\Fsat[\delta]$
    from a state $(\alpha,\mathcal L,\mathcal N)$ where $\alpha$ is total.
    The functions $f$ are assumed to be computed by machines $M_f^?$
    and $\mu_f$ to be a computable uniform modulus of continuity of $f$.
}
\label{alg:box-linearisation'-cade}
\end{algorithm}

\begin{lemma}\label{th:int:lin-correct-cade}
Let $\delta\in\mathbb Q_{>0}$
and let $S=(\alpha,\mathcal L,\mathcal N)$ be a $\delta$-\ksmt state
where $\alpha$ is total and
$\interp{\mathcal L}^\alpha=\True$. Then
\Call{nlinStep$_\delta$}{$\alpha,\mathcal L,\mathcal N$}
computes
a state derivable by application of either $(L)$ or $\Fsat[\delta]$ to $S$.
\end{lemma}

\begin{proof}
In the proof we will use notions from computable analysis, as defined in \Cref{sec:ca}.
Let $(\alpha,\mathcal L,\mathcal N)$ be a state as in the premise
and let $P:f(\vec x)\diamond0$ be a non-linear constraint in $\mathcal N$.
Let $M_f^?$ compute $f$ as in \Cref{alg:box-linearisation'-cade}.
The algorithm computes a rational approximation $\tilde y=M_f^{(\interp{\vec x}^\alpha)_i}(p)$
of $f(\interp{\vec x}^\alpha)$
where $p\geq-\lfloor\log_2(\min\{1,\delta/4\})\rfloor\in\mathbb N$.
$\interp{\mathcal L}^\alpha=\True$ implies
$\interp{\vec x}^\alpha\in D_P\subseteq\dom f$, thus the computation of $\tilde y$ terminates.
Since $M_f^?$ computes $f$, $\tilde y$ is accurate up to $2^{-p}\leq\delta/4$, that is,
$\tilde y\in[f(\interp{\vec x}^\alpha)\pm\delta/4]$.
By assumption $\diamond\in\{{>},{\geq}\}$, thus
    \begin{enumerate}
    \item\label{th:int:lin:delta-sat} $\tilde y\mathrel\diamond-\delta/2$
        implies $f(\interp{\vec x}^\alpha)\mathrel\diamond-\delta$, which is equivalent to $\interp{P_\delta}^\alpha=\True$, and
    \item\label{th:int:lin:unsat} $\neg(\tilde y\mathrel\diamond-\delta/2)$
        implies $\neg(f(\interp{\vec x}^\alpha)\mathrel\diamond-\delta/2+\delta/4)$, which in turn implies $\interp P^\alpha=\False$
        and the applicability of rule $(L)$.
    \end{enumerate}
For \cref{th:int:lin:delta-sat} no linearisation is necessary and indeed the algorithm does not linearise $P$.
Otherwise (\Cref{th:int:lin:unsat}),
it adds the linearisation $(\vec x\notin B(\interp{\vec x}^\alpha,\eps_P))$ to the linear clauses.
Since $\interp{\vec x}^\alpha\in D_P$
by \cref{eq:mu} we obtain that
$0\notin B(f(\vec z),\delta/4)$ holds,
implying $\neg(f(\vec z)\diamond 0)$,
for all $\vec z\in B(\interp{\vec x}^\alpha,\eps_P)\cap\widebar{D_P}$.
Hence,
$(\vec x\notin B(\interp{\vec x}^\alpha,\eps_P))$
is a linearisation of $P$ at $\alpha$.

In case \Call{nlinStep$_\delta$}{$\alpha,\mathcal L,\mathcal N$} returns
$\delta$-\SAT, the premise of \Cref{th:int:lin:delta-sat} holds for every non-linear constraint in $\mathcal N$, that is, $\interp{\mathcal N_\delta}^\alpha=\True$.
By assumption $\interp{\mathcal L}^\alpha=\True$, hence the application of the $\Fsat[\delta]$ rule deriving $\delta$-\SAT is possible in $\delta$-\ksmt.
\end{proof}

\begin{lemma}\label{th:int:lin-epsfull-cade}
For any bounded instance $\mathcal L_0\land\mathcal N$
there is a computable $\eps\in\mathbb Q_{>0}$ such that
any $\delta$-\ksmt
run starting in $(\nil,\mathcal L_0,\mathcal N)$, where applications of $(L)$ and $\Fsat[\delta]$ are performed
by \Call{nlinStep$_\delta$}{}, is \epsfull.
\end{lemma}
\begin{proof}
Let $P:f(\vec x)\diamond 0$ be a non-linear constraint in $\mathcal N$.
Since $\mathcal L_0\land\mathcal N$ is a bounded instance, $D_P\subseteq\mathbb R^n$ is also bounded.
Let $\eps_P\coloneqq2^{-\mu_f(p)}$
where $p\geq-\lfloor\log_2(\min\{1,\delta/4\})\rfloor\in\mathbb N$
as in \Cref{alg:box-linearisation'-cade}.
As $\mu_f$ is a uniform modulus of continuity,
the inequalities in the following construction hold on the whole domain
 $\widebar{D_P}$ of $f$ and do not depend on the concrete assignment $\alpha$ where the linearisation is performed.
Since $\log_2$ and $\mu_f$ are computable, so are $p$ and $\eps_P$.
There are finitely many non-linear constraints $P$ in $\mathcal N$,
therefore the linearisations the algorithm \Call{nlinStep$_\delta$}{}
computes are \epsfull with $\epsilon=\min\{\eps_P:P~\text{in}~\mathcal N\}>0$.
\end{proof}

We call $\delta$-\ksmt derivations when linearisation are computed using Algorithm~\ref{alg:box-linearisation'-cade}  $\delta$-\ksmt with full-box linearisations, or \emph{$\delta$-\ksmt-fb} for short.
As the runs computed by it are \epsfull for $\eps>0$, by \Cref{th:bounded} they terminate.

\begin{theorem}\label{thm:delta-ksmt-complete-cade}
$\delta$-\ksmt-fb is a $\delta$-complete decision procedure. 
\end{theorem}
\begin{proof}
$\delta$-\ksmt-fb is sound (\Cref{thm:delta_sound-cade}) and
terminates on bounded instances (\Cref{th:bounded,th:int:lin-epsfull-cade}).
\end{proof}

\section{Local \epsfull linearisations} \label{sec:eps-from-M-cade}
In practice, when the algorithm computing \epsfull linearisations described in the previous section is going to be implemented, the question arises of how to get a good uniform modulus of continuity $\mu_f$ for a computable function $f$. Depending on how $f$ is given, there may be several ways of computing it.
Implementations of exact real arithmetic, e.g.,
iRRAM~
\cite{DBLP:conf/cca/Muller00} and
Ariadne~
\cite{https://doi.org/10.1002/rnc.2914}, are usually based on the formalism of function-oracle Turing machines (see \Cref{def:computable}) which allow to compute with representations of computable functions~\cite{DBLP:journals/corr/BrausseS17} including
implicit representations of functions as solutions of
ODEs/PDEs~\cite{DBLP:books/daglib/0087495,DBLP:conf/ershov/BrausseKM15}.
If $f$ is only available as a function-oracle Turing machine $M_f^?$ computing it,
a modulus $\mu_f$ valid on a compact domain can be computed, however,
in general this
is not possible without exploring the behaviour of the function on the whole domain, which in many cases is computationally expensive.
Moreover, since $\mu_f$ is uniform, $\mu_f(n)$ is constant throughout $D_F$, independent of the actual assignment $\alpha$ determining where $f$ is evaluated.
Yet, computable functions admit \emph{local} moduli of continuity that additionally depend on the concrete point in their domain.
In most cases these would provide linearisations with $\eps$ larger than that determined by $\mu_f$ leading to larger regions being excluded,
ultimately resulting in fewer linearisation steps and general speed-up.
Indeed, machines producing finite approximations of $f(x)$ from finite approximations of $x$ internally have to compute some form of local modulus
to guarantee correctness.
In this section, we explore
this approach of obtaining linearisations covering a larger part of the function's domain.

In order to guarantee a positive
bound on the local modulus of continuity extracted directly from
the run of the machine $M_f^?$ computing $f$, it is necessary to employ a restriction on the names of real numbers $M_f^?$ computes on.
The set of
names should in a very precise sense be ``small'', i.e., it has to be compact.
The very general notion of names used in \Cref{def:computable} is too broad to satisfy this criterion since
the space of rational approximations
is not even locally compact.
Here, we present an approach using practical names of real numbers as sequences of dyadic rationals of lengths restricted by accuracy.
For that purpose, we introduce another representation~\cite{DBLP:series/txtcs/Weihrauch00} of $\mathbb R$,
that is, the surjective mapping $\xi:\mathbb D_\omega\to\mathbb R$.
Here, $\mathbb D_\omega$ denotes the set of infinite sequences $\varphi$ of dyadic rationals with bounded length. If $\varphi$ has a limit (in $\mathbb R$), we write $\lim\varphi$.

\begin{definition}\label{def:Domega-xi-Cauchy}
\begin{itemize}
\item
    For $k\in\omega$ let
    $\mathbb D_k\coloneqq\mathbb Z\cdot2^{-(k+1)}=\{m/2^{k+1}:m\in\mathbb Z\}\subset\mathbb Q$
    and let
    $\mathbb D_\omega\coloneqq\bigtimes_{k\in\omega}\mathbb D_k$
    be the set of all sequences $(\varphi_k)_k$ with $\varphi_k\in\mathbb D_k$ for all $k\in\omega$.
    By default, $\mathbb D_\omega$ is endowed with the Baire space topology,
    which corresponds to that induced by the metric
    \[ d:(\varphi,\psi)\mapsto\begin{cases}
        0&\text{if}~\varphi=\psi \\
        1/{\min\{1+n:n\in\omega,\varphi_n\neq\psi_n\}}&\text{otherwise.}
    \end{cases} \]
\item
    Define $\xi:\mathbb D_\omega\to\mathbb R$ as the partial function mapping
    $\varphi\in\mathbb D_\omega$ to
    $\lim\varphi$ iff $\forall i,j:|\varphi_i-\varphi_{i+j}|\leq2^{-(i+1)}$.
    Any $\varphi\in\xi^{-1}(x)$ is called a \emph{$\xi$-name} of $x\in\mathbb R$.
\item
    The representation $\rho:(x_k)_k\mapsto x$ mapping names $(x_k)_k$ of $x\in\mathbb R$ to $x$ as per \Cref{def:computable} is called \emph{Cauchy representation}.
\end{itemize}
\end{definition}
Using a standard product construction we can easily generalise the notion of $\xi$-names to $\xi^n$-names of $\mathbb R^n$.
When clear from the context, we will drop $n$ and just write $\xi$ to denote the corresponding generalised representation $\mathbb D_\omega^n\to\mathbb R^n$.

Computable equivalence between two representations not only implies that there are continuous maps between them but also that names can computably be transformed~\cite{DBLP:series/txtcs/Weihrauch00}.
Since the Cauchy representation itself is continuous~\cite{DBLP:journals/tcs/BrattkaH02} we
derive continuity of $\xi$,
which is used below to show compactness of preimages $\xi^{-1}(X)$ of compact sets $X\subseteq\mathbb R$ under $\xi$.
\cade{%
    All proofs can be found in~\todo[inlinepar]{\cite{}}.
}{%
    All proofs can be found in the appendix.
}

\begin{lemma}\label{prop:cade}
The following properties hold for $\xi$.
\begin{enumerate}
\item $\xi$ is a representation of $\mathbb R^n$:
    it is well-defined and surjective.
\item\label{it:xi2rho-cade}
    Any $\xi$-name of $\vec x\in\mathbb R^n$ is a Cauchy-name of $\vec x$.
\item\label{it:xi-equiv-rho-cade}
    $\xi$ is computably equivalent to the Cauchy representation.
\item $\xi$ is continuous.
\end{enumerate}
\end{lemma}%
The converse of \cref{it:xi2rho-cade} does not hold. An example for a Cauchy-name of $0\in\mathbb R$ is the sequence $(x_n)_n$ with $x_n=(-2)^{-n}$ for all $n\in\omega$, which does not satisfy $\forall i,j:|x_i-x_{i+j}|\leq2^{-(i+1)}$.
However, given a name of a real number, we can compute a corresponding $\xi$-name, this is one direction of the property in \cref{it:xi-equiv-rho-cade}.

As a consequence of \cref{it:xi2rho-cade} a function-oracle machine $M^?$ computing $f:\mathbb R^n\to\mathbb R$ according to \Cref{def:computable} can be run on $\xi$-names of $\vec x\in\mathbb R^n$ leading to valid Cauchy-names of $f(\vec x)$.
Note that this proposition does not require $M_f^?$ to compute a $\xi$-name of $f(\vec x)$.
Any rational sequence rapidly converging to $f(\vec x)$ is a valid output.
This means, that the model of computation remains unchanged with respect to the
earlier parts of this paper.
It is the set of names the machines are operated on, which is restricted.
This is reflected in \Cref{alg:2}
by computing dyadic rational approximations $\tilde{\vec x}_k$ of $\interp{\vec x}^\alpha$ such that
$\tilde{\vec x}_k\in\mathbb D_k^n$ instead of keeping the name of $\interp{\vec x}^\alpha$ constant
as has been done in \Cref{alg:box-linearisation'-cade}.

\begin{algorithm}[tp]
\begin{algorithmic}
\Function{LineariseLocal$_\delta$}{$f,\vec x,\diamond,\alpha$}
    \State $\varphi\gets(m\mapsto \Approx(\interp{\vec x}^\alpha,m))$
    \Comment then $\varphi$ is a $\xi$-name of $\interp{\vec x}^\alpha$
    \State compute $p\geq-\lfloor\log_2(\min\{1,\delta/4\})\rfloor$
    \State run $M_f^\varphi(p+2)$, record its output $\tilde y$ and its maximum query $k\in\omega$ to $\varphi$
    \If{$\tilde y\diamond-\delta/2$}
        \State \Return $\None$
    \Else
        \State \Return $(\vec x\notin B(\interp{\vec x}^\alpha,2^{-k}))$
    \EndIf
\EndFunction
\end{algorithmic}
\caption{\textbf{(Local linearisation)}
Algorithm $\delta$-deciding $P:f(\vec x)\diamond 0$ and -- in case \UNSAT{} -- computing a linearisation at $\alpha$ or returning ``\None'' and in this case $\alpha$ satisfies $P_\delta$.
The function $f$ is computed by machine $M_f^?$.
}
\label{alg:2}
\end{algorithm}

In particular, in \Cref{th:int:lin:alt-cade} we show that linearisations for the $(L_\delta)$ rule can be computed by \Cref{alg:2}, which -- in contrast to \Call{linearise$_\delta$}{} in \Cref{alg:box-linearisation'-cade} -- does not require access to a procedure computing an upper bound $\mu_f$ on the uniform modulus of continuity of the non-linear function $f\in\mathcal F_{\mathrm{nl}}$ valid on the entire bounded domain.
It not just runs the machine $M_f^?$, but also
observes the queries $M_f^\varphi$ poses to its oracle in order to
obtain a local modulus of continuity of $f$ at the point of evaluation.
The function $\Approx(\vec x,m)\coloneqq\round{\vec x\cdot 2^{m+1}}/2^{m+1}$ used to define \Cref{alg:2} computes a dyadic approximation of $\vec x$, with
$\round\cdot:\mathbb Q^n\to\mathbb Z^n$ denoting a rounding operation, that is, it satisfies 
$\forall \vec q:\Vert\round{\vec q}-\vec q\Vert\leq\frac12$. On rationals (our use-case), $\round\cdot$ is computable by a classical Turing machine.

\begin{definition}[{\cite[Definition 6.2.6]{DBLP:series/txtcs/Weihrauch00}}]\label{def:local-mod}
Let $f:\mathbb R^n\to\mathbb R$ and $\vec x\in\dom f$.
A function $\gamma:\mathbb N\to\mathbb N$ is called \emph{a (local) modulus of continuity of $f$ at $\vec x$} if for all $p\in\mathbb N$ and $\vec y\in\dom f$,
$\Vert\vec x-\vec y\Vert\leq 2^{-\gamma(p)}\implies\vert f(\vec x)-f(\vec y)\vert\leq2^{-p}$
holds.
\end{definition}
We note that in most cases a local modulus of continuity of $f$ at $\vec x$
is smaller than the best uniform modulus of $f$ on its domain, since it only depends on the local behaviour of $f$ around $x$.
One way of computing a local modulus of $f$ at $\vec x$ is using the function-oracle machine $M_f^?$ as defined next.
\begin{definition}\label{def:local-mod-M}
Let $M^?_f$ compute $f:\mathbb R^n\to\mathbb R$ and let $\vec x\in\dom f$ have Cauchy-name $\varphi$.
The function
$\gamma_{M_f^?,\varphi}:p\mapsto\max\{0,k:k~\text{is queried by}~M_f^\varphi(p+2)\}$
is called
\emph{the effective local modulus of continuity induced by $M_f^?$ at $\varphi$}.
\end{definition}
The effective local modulus of continuity of $f$ at a name $\varphi$ of $\vec x\in\dom f$ indeed is a local modulus of continuity of $f$ at $\vec x$~\cite[Theorem~2.13]{DBLP:books/daglib/0067010}.

We prove that \Cref{alg:2} indeed computes linearisations
in \Cref{prf:local-lin''}.

\begin{lemma}\label{lem:local-lin''}
Let $P:f(\vec x)\diamond0$ be a non-linear constraint in $\mathcal N$ and
$\alpha$ be an assignment of $\vec x$ to rationals in $\dom f$.
Whenever $C={}$\Call{LineariseLocal$_\delta$}{$f,\vec x,\diamond,\alpha$}
and $C\neq\None$, $C$ is an \epsfull linearisation of $P$ at $\alpha$, 
 with $\epsilon$
corresponding to the effective local modulus of continuity
induced by $M_f^?$ at a $\xi$-name of $\interp{\vec x}^\alpha$.
\end{lemma}

Thus, the function \textsc{lineariseLocal$_\delta$}{} in \Cref{alg:2}
is a drop-in replacement for \textsc{linearise$_\delta$}{} in \Cref{alg:box-linearisation'-cade} since the condition on returning a linearisation of $P$ versus accepting $P_\delta$ is identical.
The linearisations however differ in the radius $\epsilon$, which now, according to \Cref{lem:local-lin''}, corresponds to the effective local modulus of continuity.
The resulting procedure we call \Call{nlinStepLocal$_\delta$}{}.
One of its advantages over \Call{nlinStep$_\delta$}{} is
running $M_f^?$ on $\xi$-names instead of Cauchy-names, is that they form a compact set for bounded instances, unlike the latter.   This allows us
to bound $\eps>0$ for the computed \epsfull local linearisations of otherwise arbitrary $\delta$-\ksmt runs.
A proof of the following Lemma showing
compactness of preimages $\xi^{-1}(X)$ of compact sets $X\subseteq\mathbb R$ under $\xi$ is given in \Cref{prf:names-compact-cade}.
\begin{lemma}\label{prop:names-compact-cade}
    Let $X\subset\mathbb R^n$ be compact.
    Then the set $\xi^{-1}(X)\subset\mathbb D_\omega^n$ of
    $\xi$-names of elements in $X$ is compact as well.
\end{lemma}
The proof involves showing $\xi^{-1}(X)$ to be closed and uses the fact that for each component $\varphi_k$ of names $(\varphi_k)_k$ of $\vec x\in X$ there are just finitely many choices from $\mathbb D_k$ due to the restriction of the length of the dyadics. This is not the case for the Cauchy representation used in \Cref{def:computable} and
\todorev[2]{`sceptical'. \emph{`We will [...] add clarifications'}}%
it is the key for deriving existence of a strictly positive lower bound $\epsilon$ on the \epsfull{}ness of linearisations.

\begin{theorem}\label{th:int:lin:alt-cade}
Let $\delta\in\mathbb Q_{>0}$.
For any bounded instance $\mathcal L_0\land\mathcal N$
there is $\epsilon>0$ such that any $\delta$-\ksmt run starting in $(\nil,\mathcal L_0,\mathcal N)$, where applications of $(L)$ and $\Fsat[\delta]$ are performed according to \Call{nlinStepLocal$_\delta$}{}, is \epsfull.
\end{theorem}
\begin{proof}
Assume $\mathcal L_0\land\mathcal N$ is a bounded instance.
Set $\eps\coloneqq\min\{\eps_P:P\in\mathcal N\}$, where $\eps_P$ is defined as follows.
Let $P:f(\vec x)\diamond 0$ in $\mathcal N$.
Then the closure $\widebar{D_P}$ of the
bounded set $D_P$ is compact.
Let $E$ be the set of $\xi$-names of elements of $\widebar{D_P}\subseteq\dom f$ (see \Cref{def:Domega-xi-Cauchy})
and for any $\varphi\in E$
let $k_\varphi$ be the maximum index queried by $M_f^\varphi(p)$ where $p$ is computed from $\delta$ as in \Cref{alg:2}.
Therefore $\varphi\mapsto k_\varphi$ is continuous.
By \Cref{prop:names-compact-cade} $E$ is compact, thus,
there
\todofb{could also state $k_\psi=\max\{k_\varphi:\varphi\in E\}$ instead; which is clearer? \fb{keep} \mk{keep}}%
is $\psi\in E$ such that $2^{-k_\psi}=
\inf\{2^{-k_\varphi}:\varphi\in E\}$.
Set $\eps_P\coloneqq2^{-k_\psi}$.
The claim then follows by \Cref{lem:local-lin''}.
\end{proof}

Thus we can conclude.
\begin{corollary}
 $\delta$-\ksmt with local linearisations is a $\delta$-complete decision procedure.
\end{corollary}

\section{Conclusion}\label{sec:concl}
In this paper we extended the the \ksmt calculus to the $\delta$-satisfiability setting and proved that  
the resulting $\delta$-\ksmt calculus is a $\delta$-complete decision procedure for solving non-linear constraints over computable functions which include polynomials, exponentials, logarithms, trigonometric and
many other functions used in applications. 
We presented algorithms for constructing $\epsilon$-full linearisations
ensuring termination of $\delta$-\ksmt.
Based on methods from computable analysis we presented an algorithm for constructing local linearisations. Local linearisations exclude larger regions from the search space and can be used to avoid computationally expensive global analysis of non-linear functions. 

\bibliographystyle{plain}
\bibliography{references}

\cade{}{
\appendix

\section{Proofs}
\subsection{Proof of \Cref{prop:cade}}

    \subsubsection{Any $\xi$-name of $\vec x\in\mathbb R^n$ is a Cauchy-name of $\vec x$}
    \begin{lemma}\label{lem:dist-cade}
    For any $\vec x\in\mathbb R^n$ and $\xi$-name $\varphi$ of $\vec x$, $\forall k:|\vec x-\varphi_k|\leq 2^{-k}$ holds.
    \end{lemma}
    \begin{proof}
    For simplicity we assume a dimension of $n=1$.
    The general case can be proved similarly.
    Let $x\in\mathbb R$ and $\varphi$ be a $\xi$-name of $x$
    and let $k\in\omega$.
    By construction $x=\lim\varphi$, hence there is
    $n_0\in\omega$ such that for every $n\geq n_0$ the bound
    $|\varphi_n-x|<2^{-k-1}$ holds.
    If $n_0\leq k$, the previous bound already gives the required property.
    Otherwise $n_0>k$, then
    $|\varphi_k-x|\leq|\varphi_k-\varphi_{n_0}|+|\varphi_{n_0}-x|$
    holds. Since $\varphi\in\dom\xi$, the first summand is bounded by $2^{-\min(k,n_0)-1}=2^{-(k+1)}$.
    By the property above, so is the second.
    Ergo $|\varphi_k-x|\leq2^{-k}$.
    \end{proof}
    The property that $\xi$-names are Cauchy-names follows directly from \Cref{lem:dist-cade}.
    
    \subsubsection{$\xi$ is computably equivalent to the Cauchy representation}
    \begin{proof}
    For simplicity we assume a dimension of $1$. The general case can be proved similarly.
    \begin{itemize}
    \item[$\Rightarrow$)]
        Let $\psi$ be a $\xi$-name of $x\in\mathbb R$.
        By \Cref{it:xi2rho-cade} of \Cref{prop:cade}, $\psi$
        is a
        name of $x$.
    \item[$\Leftarrow$)]
        Given $\varphi\in\dom\rho$ and $n\in\omega$.
        Compute $\psi_n\coloneqq\round{\varphi_{n+4}\cdot2^{n+1}}/2^{n+1}\in\mathbb D_n$
        where $\round\cdot:\mathbb Q\to\mathbb Z$ is a computable rounding operation
        with $|\round q-q|\leq1/2$.
    
        Then with $x\coloneqq\lim\varphi$:
        \begin{align*}
    	|\psi_n-x|&=|\round{\varphi_{n+4}\cdot2^{n+1}}/2^{n+1}-x| \\
    	&\leq|\round{\varphi_{n+4}\cdot2^{n+1}}-\varphi_{n+4}\cdot2^{n+1}|/2^{n+1}+
    	     |\varphi_{n+4}-x| \\
    	&\leq2^{-(n+2)}+2^{-(n+4)}
        \end{align*}
    
        We show $\psi\coloneqq(\psi_n)_n$ is a $\xi$-name of $x$.
        Let $n,k\in\omega$ with $k>0$.
        \begin{align*}
    	|\psi_n-\psi_{n+k}|
    	&\leq|\psi_n-x|+|\psi_{n+k}-x| \\
    	&\leq2^{-(n+2)}+2^{-(n+4)}+2^{-(n+k+2)}+2^{-(n+k+4)} \\
    	&\leq2^{-(n+2)}+2^{-(n+4)}+2^{-(n+3)}+2^{-(n+5)} \\
    	&\leq2^{-(n+1)}
        \end{align*}
        Thus, $\psi\in\dom\xi$ and
        therefore $\psi$ is a $\xi$-name of $\lim\psi=x$.
    \end{itemize}
    \end{proof}
    
    \subsubsection{$\xi$ is continuous}
    \begin{proof}
    Computable equivalence between two representations implies there are continuous maps between them.
    Since the Cauchy representation is continuous itself~\cite{DBLP:journals/tcs/BrattkaH02}, so is $\xi$.
    \end{proof}

\subsection{Proof of \Cref{lem:local-lin''}}\label{prf:local-lin''}

\paragraph{\upshape\textbf{\Cref{lem:local-lin''}.}} \emph{%
Let $P:f(\vec x)\diamond0$ be a non-linear constraint in $\mathcal N$ and
$\alpha$ be an assignment of $\vec x$ to rationals in $\dom f$.
Whenever $C={}$\Call{LineariseLocal$_\delta$}{$f,\vec x,\diamond,\alpha$}
and $C\neq\None$, $C$ is an \epsfull linearisation of $P$ at $\alpha$, 
 with $\epsilon$
corresponding to the effective local modulus of continuity
induced by $M_f^?$ at a $\xi$-name of $\interp{\vec x}^\alpha$.
}

\begin{proof}
Let $P:f(\vec x)\diamond 0$, $\alpha$ and $C\neq\None$ be as in the premise and
let $p,\tilde y$ and $\varphi$ as in \Cref{alg:2}.
Since $C\neq\None$ by construction
$C=(\vec x\notin B(\interp{\vec x}^\alpha,2^{-k}))$ where
$k=\gamma_{M^?_f,\varphi}(p)$ is the maximum query $M_f^\varphi(p+2)$ poses to its oracle.
Thus, by definition, $C$ is \epsfull with $\epsilon=2^{-k}$.
In order to show that $C$ indeed is a linearisation of $P$ at $\alpha$,
let $\vec z\in B(\interp{\vec x}^\alpha,2^{-k})\cap\widebar{D_P}\subseteq\dom f$.
As $\gamma_{M^?_f,\varphi}$ is a local modulus of continuity of $f$ at $\interp{\vec x}^\alpha$~\cite[Theorem~2.13]{DBLP:books/daglib/0067010},
$f(\vec z)$ is within distance $2^{-p}\leq\delta/4$ of $f(\interp{\vec x}^\alpha)$, which, by definition of $M_f^?$, is at most $2^{-(p+2)}<\delta/4$ away from $\tilde y$.
By construction of $C$ in \Cref{alg:2}
the property $\neg(\tilde y\diamond-\delta/2)$ holds.
As in case~\ref{th:int:lin:unsat} in the proof of \Cref{th:int:lin-correct-cade},
this property implies
$\neg(f(\vec z)\diamond 0)$.
Therefore, according to \Cref{def:eps-full}, $C$ is an \epsfull linearisation of $P$ at $\alpha$.
\end{proof}

\subsection{Proof of \Cref{prop:names-compact-cade}}\label{prf:names-compact-cade}

The proof of the following lemma
follows that of~\cite[Theorem 7.2.5.2]{DBLP:series/txtcs/Weihrauch00}.

\begin{lemma}\label{prop:names-closed-cade}
Let $X\subset\mathbb R^n$ be closed.
Then the set $\xi^{-1}(X)\subset\mathbb D_\omega^n$ of
$\xi$-names of elements in $X$ is closed as well.
\end{lemma}
\begin{proof}
Again, for sake of simplicity we assume $n=1$ while the
general case be proved in a similar manner.

We first introduce the notation $\mathbb D_*$ for the
set of finite prefixes of elements in $\mathbb D_\omega$
and for any prefix $u\in\mathbb D_*$ let
$\mathbb D_\omega[u]\coloneqq\{u\varphi\mid \varphi\in\omega^\omega,u\varphi\in\mathbb D_\omega\}$ denote `the
ball' around $u$ in $\mathbb D_\omega$.

Observe that $\mathbb D_\omega[u]$ is a basic open set
in $\mathbb D_\omega$ for any
prefix $u\in\mathbb D_*$.

In order to show $\xi^{-1}(X)$ is closed, we prove that
its complement is open.
Since $X$ is closed there is a
\KK{$\mathcal B$ not needed \fb y}%
collection $\mathcal B$ of open
subsets of $\mathbb R$ such that $\bigcup\mathcal B=\mathbb R\setminus X$. Then
\[ C\coloneqq\{u\in\omega^*\mid\exists B\in\mathcal B~\text{s.t.}~
   \xi(\mathbb D_\omega[u])\subseteq B\}
\]
is a subset of $\omega^*$. Define
\begin{align*}
    U_1&\coloneqq\bigcup_{u\in C}\mathbb D_\omega[u] \\
    U_2&\coloneqq\{(p_n)_n\in\mathbb D_\omega\mid
        \exists i,j~\text{s.t.}~|p_i-p_{i+j}|>2^{-(i+1)}\}
\end{align*}
and $U\coloneqq U_1\cup U_2$.
By definition, $U_1$ is open in $\mathbb D_\omega$.
To see that $U_2$ is open in $\mathbb D_\omega$ as well,
observe that
\[ U_2=\mathbb D_\omega\setminus\dom\xi=\bigcup_{w\in C'}\mathbb D_\omega[w] \]
where $w\in C'\subseteq\mathbb D_*$ iff there are $i,j\in\omega$ such that $|w_i-w_{i+j}|>2^{-(i+1)}$.

We show $\xi^{-1}(X)=\mathbb D_\omega\setminus U$.
\begin{itemize}
\item Assume $p\in\xi^{-1}(X)$, that is, $\xi(p)\in X$.
	\begin{itemize}
	\item 
		Then $p=(p_k)_k$ with $p_k\cdot2^{k+1}\in\mathbb Z$ for every $k$
		such that $\forall i,j:|p_i-p_{i+j}|\leq2^{-(i+1)}$, thus, $p\notin U_2$.
	\item 
		Suppose $p\in U_1$.
		Then there is $u\in C$ and $B\in\mathcal B$ such that
		$p\in\mathbb D_\omega[u]$ and $\xi(\mathbb D_\omega[u])\subseteq B$ or
		$\mathbb D_\omega[u]\cap\dom\xi=\varnothing$.
		In both cases $p\notin\xi^{-1}(X)$, a contradiction.
	\end{itemize}

	From $p\notin U_1$ and $p\notin U_2$ we obtain $p\notin U$.

\item Assume $p\notin\xi^{-1}(X)$.

	\begin{itemize}
	\item First, consider $p\notin\dom\xi$. Then $p\in U_2$.

	\item Finally, consider $p\in\dom\xi$.
	    Since $\xi(p)\notin X$, there are
	    \KK{elaborate: whole prefix will be outside \fb y}%
		some prefix $u$ of $p$ and
		some $B\in\mathcal B$ such that $\xi(\mathbb D_\omega[u])\subseteq B$,
		and so $p\in U_1$.
	\end{itemize}
	From $p\in U_1$ or $p\in U_2$ we obtain $p\in U$.
\end{itemize}
Therefore, $U$ is the open complement of $\xi^{-1}(X)$, which is a closed subset of $\mathbb D_\omega$.
\end{proof}

Now, the proof of \Cref{prop:names-compact-cade} follows from Tychonoff's theorem,
which states that arbitrary products of non-empty compact spaces again are compact~\cite{Tychonoff1930}.

\paragraph{\upshape\textbf{\Cref{prop:names-compact-cade}.}} \emph{%
    Let $X\subset\mathbb R^n$ be compact.
    Then the set $\xi^{-1}(X)\subset\mathbb D_\omega^n$ of
    $\xi$-names of elements in $X$ is compact as well.
}
\begin{proof}
By \Cref{lem:dist-cade},
$\xi^{-1}(X)$ is a subset of the product of the
finite and therefore compact spaces
\[ \{\vec y\in\mathbb D_k^n:\Vert\vec x-\vec y\Vert\leq2^{-k},\vec x\in X\} \] over $k\in\omega$.
As a closed (\Cref{prop:names-closed-cade}) subset of a compact space, $\xi^{-1}(X)$ is compact as well.
\end{proof}
}

\end{document}